\documentclass[11pt,fleqn]{article}

\usepackage{amsmath,amssymb,amsthm,enumerate,hyperref} 

\newcommand{\todo}[1][\null]{\ensuremath{\clubsuit}}

\newcommand{\noprint}[1]{}

\newcommand{\sign}{\mathop{\rm sgn}\nolimits}

\newcounter{mcasenum}

\newtheorem{theorem}{Theorem}

\newtheorem*{proposition*}{Proposition}
{\theoremstyle{definition}

}

\begin{document}
\allowdisplaybreaks

\begin{center}{\Large{\bf
Classification of reduction operators and exact solutions \\[1ex]of variable coefficient Newell--Whitehead--Segel equations}}

{\vspace{4mm}\par\noindent\large Olena Vaneeva$^\dag$, Vyacheslav Boyko$^\dag$, Alexander Zhalij$^\dag$ and Christodoulos Sophocleous$^\ddag$}
\end{center}

\par\noindent{$^{\dag}${\it
 Institute of Mathematics of NAS of Ukraine,
 3 Tereshchenkivs'ka Str., 01004 Kyiv, Ukraine }\\
$^{\ddag}${\it
 Department of Mathematics and Statistics, University of Cyprus, Nicosia CY 1678, Cyprus}
 \\[1ex]
\centerline{\href{mailto:vaneeva@imath.kiev.ua}{vaneeva@imath.kiev.ua},\quad \href{mailto:boyko@imath.kiev.ua}{boyko@imath.kiev.ua},\quad \href{mailto:zhaliy@imath.kiev.ua}{zhaliy@imath.kiev.ua},\quad\href{mailto:christod@ucy.ac.cy}{christod@ucy.ac.cy}}}

{\vspace{6mm}\par\noindent\hspace*{8mm}\parbox{146mm}{\small
A class of the Newell--Whitehead--Segel equations (also known as generalized Fisher equations and Newell--Whitehead equations) is studied with Lie and ``nonclassical'' symmetry points of view. The classifications of Lie reduction operators and of regular nonclassical reduction operators are performed. The set of admissible transformations (the equivalence groupoid) of the class is described exhaustively. The criterion of reducibility of variable coefficient Newell--Whitehead--Segel equations to their constant coefficient counterparts is derived. Wide families
of exact solutions for such variable coefficient equations are constructed.

}\par}

\section{Introduction}
A reduction operator of a $(1{+}1)$-dimensional partial differential equation (PDE) with independent variables $t$ and $x$ and the dependent variable $u$ is a differential operator of the form $
Q = \tau(t,x,u)\partial_t + \xi(t,x,u)\partial_x + \eta(t,x,u)\partial_u$,
$(\tau,\xi)\ne(0,0)$,
such that the corresponding invariant surface condition $Q[u]:=\tau u_t+\xi u_x-\eta=0$ leads to the construction of ansatz that reduces the number of independent variables of the respective equation by one.
Thus, such operators allows one to reduce a $(1{+}1)$-dimensional PDE to an ordinary differential equation.

\looseness=-1 The reduction method is an efficient tool for seeking exact solutions of nonlinear PDEs as the general theory of integration of such equations does not exist. Among the most known reduction techniques are the prominent Lie reduction method that originates from works by S.~Lie and the nonclassical reduction method suggested by G.W. Bluman in~\cite{Bluman1968} (see also~\cite{Bluman&Cole1969}). The criterion of ``nonclassical'' invariance was firstly formulated in~\cite{Fushchych&Tsyfra1987} and the rigorous theory of the nonclassical reduction method, \emph{theory of reduction modules}, was recently developed in~\cite{Boyko}. The nonclassical reduction operators are also called {\it nonclassical symmetries}~\cite{OlverRosenau1987}, {\it conditional symmetries}~\cite{levi1989} and {\it$Q$-conditional symmetries}~\cite{FushchichShtelenSerov-book} (see the related discussion in~\cite{KunzingerPopovych2009} and some more research papers of interest~\cite{Fushchych92,Fushchych95,GrundlandTafel1995,Nucci,OlverVorobev1996,yehorchenko}).

There is also a direct reduction method based on substitution of ansatz into a PDE in question~\cite{ClarksonKruskal1989, Fushchych81}.
A rigorous definition of reduction of PDEs was presented in~\cite{ZhdanovTsyfraPopovych1998}.
It was proved therein that the direct approach of reduction, taken in its
full generality, is equivalent to the non-classical (conditional symmetry) approach. The enhanced proof can be found in~\cite{Boyko}.

Therefore an important problem arises: to classify reduction operators for those classes of PDEs that are of interest for applications.
Classification of Lie reduction operators is known as {\it group classification problem} and appears to be the central problem of the group analysis. The main benefit of Lie method is that the determining system for finding coefficients of Lie reduction operators consists of linear PDEs. That is why the construction of Lie symmetry operators for a fixed PDE is a routine task usually which can be performed using the packages of symbolic computations.
See, for example, the Maple-based GEM package~\cite{Cheviakov,Cheviakov2010}. Unfortunately the group classification problems can be solved automatically using symbolic computations only for certain classes having simple structures. The majority of cases requires usage of the modern techniques of the group analysis such as mapping between classes of PDEs, gauging of arbitrary elements of the class in order to reduce their number, application of various types of equivalence groups, etc.\ (see, e.g.,~\cite{Popovych&Ivanova2004NVCDCEs,VJPS2007,VPS2009}).

The nonclassical reduction operators can be of regular and singular types. The problem of finding singular reduction operators reduces to solving an initial PDE, therefore this case is called the ``no-go'' case and often omitted in consideration (see more about ``no-go'' case in~\cite{Boyko,Fushchych&Shtelen&Serov&Popovych1992,KunzingerPopovych,Popovych2006b,Popovych2008,Zhdanov&Lahno1998}). But even in the case of regular nonclassical reduction operators
the problem of their classification for classes of PDEs is difficult. This is due to the fact that finding coefficients of nonclassical reduction operators one requires to solve a system of nonlinear PDEs.
That is why this method more often results in the complete solution when applied to a~fixed PDE rather than to a class of PDEs. Indeed, there are quite few examples of successful classification of nonclassical reduction operators (even regular ones) in the literature. At the best of our knowledge, such classifications are performed for the class of
semilinear diffusion equations with a source $u_t=u_{xx}+f(u)$~\cite{ArrigoHillBroadbridge1993,Clarkson&Mansfield1993,FushchichSerov1990}, the class of nonlinear reaction--diffusion equations $u_t=(D(u)u_x)_x+f(u)$ for the cases of exponential and power low diffusivity~\cite{ArrigoHill1995}, the class of nonlinear filtration equations $u_t=f(u_x)u_{xx}$~\cite{Popovych&VaneevaIvaanova2007}, the class of variable coefficient Huxley equations
$u_t=u_{xx}+k(x)u^2(1-u)$~\cite{Bradshaw2007, Ivanova&Sophocleous2010}, and the class of generalized Burgers equations $u_t=uu_x +f(t,x)u_{xx}$~\cite{PocheketaPopovych2017}.

We aim to perform exhaustive classifications of Lie and regular nonclassical reduction operators
 for the class of equations of the form
\begin{equation}\label{eq_fish}
u_t=a^2(t)u_{xx}+b(t)u-c(t)u^3,
\end{equation}
where $a(t)$, $b(t)$ and $c(t)$ are arbitrary smooth functions, $a(t)$ and $c(t)$ are nonvanishing.
This is a class of variable coefficient Newell--Whitehead--Segel equations called also in the literature generalized Fisher equations and Newell--Whitehead equations.

The classical Newell--Whitehead--Segel equation, $u_t=u_{xx}+u-u^3$, was derived in~\cite{NewellWhitehead,Segel} and it is particular case of generalized Fisher equations
\begin{equation}\label{eq_fish_gen}
u_t=\big(u^mu_x\big)_x+u^p\big(1-u^q\big),
\end{equation}
which appear as insect and animal dispersal and
invasion models in the mathematical biology (cf.\ equation~(13.40) in~\cite{murray2002}).
Here $t$ and $x$ are time and spatial coordinates, respectively, $u$ is a~population density, $p$, $q$ and $m$ are positive parameters. There are also a number of models with $m=0$, which correspond to the case of density-independent diffusion. If $m=0$ and $p=q=1$, then equation~\eqref{eq_fish_gen} becomes classical Fisher equation that was originally derived in~\cite{Fisher1937} to model the propagation of a gene in a population.
Later it was proposed to consider generalized Fisher equations
with time-dependent diffusion coefficients $u_t=f(t)u_{xx}+g(t)u(1-u)$.
In practice these coefficients could represent long term changes in climate or short term seasonality~\cite{Hammond&Bortz2011, OgunKart2007}. The group classification of the latter class was carried out in~\cite{VPS2013}.

The equations~\eqref{eq_fish} with $b(t)=c(t)=1$ were studied in~\cite{OgunKart2007} using the truncated Painlev\'e expansion method in order to construct their exact solutions. Having the same goal
the whole class~\eqref{eq_fish} was considered recently in~\cite{triki}.
It appears that all the found in~\cite{triki} ``solutions'' are stationary ones and moreover do not satisfy the respective equations due to wrong signs of constants appearing therein.
We aim to construct exact solutions for equations~\eqref{eq_fish} and also to present the complete classifications of not only Lie reduction operators but also regular nonclassical ones.
We note that the group classification for the general class of (1+1)-dimensional second-order quasilinear evolution equations $u_t=F(t,x,u,u_x)u_{xx}+G(t,x,u,u_x),$ $F\neq0,$ that contains class~\eqref{eq_fish} as subclass was
performed in~\cite{Basarab2001}. Nevertheless those results obtained up to a very wide equivalence group seem to be
inconvenient to derive group classification for class~\eqref{eq_fish}.

The structure of the paper is as follows. In Section~\ref{section2} we study the transformational properties of class~(\ref{eq_fish}) in order to reduce the number of its arbitrary elements by point transformations. The criterion of reducibility of variable coefficient equations from class~(\ref{eq_fish}) to constant coefficient equations from the same class is also derived therein. Classifications of Lie and nonclassical reduction operators are carried out in Sections~\ref{section3} and~\ref{section4}, respectively. Section~\ref{section5} is devoted to the construction of exact solutions of variable coefficients Newell--Whitehead--Segel equations using the equivalence transformations.

\section{Equivalence groupoid}\label{section2}

Point transformations can essentially simplify the classification problems for classes of differential equations. So
we aim to describe firstly all point transformations each of which connects a~pair of equations from the class~\eqref{eq_fish}.
Such transformations are called form-preserving~\cite{Kingston&Sophocleous1998} or admissible~\cite{Popovych2006} or allowed transformations~\cite{Winternitz1992}. The classifications for such transformations for various classes of PDEs were carried out, in particular, in~\cite{GagnonWinternitz1993,GungorWinternitz2004,Kingston1991,OpanasdenkoBihloPopovych2017,Popovych&Kunzinger&Eshraghi2010,VJPS2007,VPS2009,VaneevaPosta2017}.
Ordered triplets, consisting of the initial and target equations and the transformations linking them, together with the operation of composition of transformations have the groupoid structure. Such a groupoid is called {\it equivalence groupoid}~\cite{PopovychBihlo}.

Using the direct method we deduce that the equivalence groupoid of class~\eqref{eq_fish} is generated
by the usual {\it equivalence transformations} from this class. These are nondegenerate point transformations which preserve the form of any equation from a given class transforming only the form of its arbitrary element(s)~\cite[pp.~64--66]{Ovsiannikov1982}.
Therefore, class~\eqref{eq_fish} is normalized (see \cite{Popovych2006,
Popovych&Kunzinger&Eshraghi2010} for related definitions).
The following statement is true.
\begin{theorem}\label{theorem:equivalence-group}
The equivalence group~$G^\sim$ of class~\eqref{eq_fish} is formed by the transformations
\begin{gather}\nonumber
\tilde t=\theta(t),\quad \tilde x=\delta_1x+\delta_2, \quad
\tilde u=\varphi(t)u, \\\label{eq_theor}
\tilde a^2(\tilde t)=\dfrac{\delta_1{}^2}{\theta_t}a^2(t),\quad
\tilde b(\tilde t)=\dfrac{1}{\varphi\theta_t}(\varphi b(t)+\varphi_t), \quad
\tilde c(\tilde t)=\dfrac{1}{\varphi^2\theta_t}c(t),\quad
\end{gather}
where $\delta_1$ and $\delta_2$ are arbitrary constants with $\delta_1\not=0$,
and the functions $\theta(t)$ and $\varphi(t)$ are arbitrary smooth functions with $\theta_t\varphi\not=0$.

These transformations generate the equivalence groupoid of class~\eqref{eq_fish}.
\end{theorem}
\begin{proof}It is proven in~\cite[Theorem~1]{VPS2014} that
 the class of $n$th-order evolution equations of the~form
\begin{gather*}
u_t=F(t)u_n+G(t,x,u_0,u_1,\dots,u_{n-1}),\quad F\ne0,\quad G_{u_iu_{n-1}}=0,\ i=1,\dots,n-1,
\end{gather*}
where $n\geqslant 2$, is normalized in the usual sense. The components for $t$, $x$, and $u$ of
the transformations which constitute its usual point equivalence group have the form
\begin{gather*}
\tilde t=T(t),\quad \tilde x=X^1(t)x+X^0(t),\quad \tilde u=U^1(t,x)u+U^0(t,x),
\end{gather*}
where $T=T(t)$, $X=X^i(t)$, and $U^i=U^i(t,x)$, $i=0,1$,
are arbitrary smooth functions of their arguments and $T_tX^1U^1\neq0$.
Since this class is normalized, equivalence groupoid of any of its subclass forms a subgroupoid in the equivalence groupoid of the whole class. We substitute $n=2$, $F=a^2(t)$ and $G=b(t)u-c(t)u^3$ into equations (5) in
\cite[Theorem~1]{VPS2014}. The resulting equations are $\tilde a^2(\tilde t)={{(X^1)}^2}a^2(t)/{\theta_t}$ and
\begin{gather*}\tilde c(\tilde t)T_t\left(U^1u+U^0\right)^3-\tilde b(\tilde t)T_t\left(U^1u+U^0\right)+U^1\left(b(t)u-c(t)u^3\right)+U^1_tu+U^0_t\\
\quad{} -\left(U^1_{xx}u+2U^1_xu_x+U^0_{xx}\right)a^2(t)-\frac{X^1_tx+X^0_t}{X^1}\left(U^1_xu+U^1u_x+U^0_x\right)=0.
\end{gather*}
We split the latter equation
with respect to $u_x$ and $u$. This provides a system of five determining equations. The coefficient of $u^2$ implies that $U^0(t,x)=0$. From this results, we note that the coefficient independent of $u$ and $u_x$ also vanishes. From the coefficient of $u^3$ we have $U^1_x=0$, so $U^1=U^1(t)$, and the relation between $\tilde c$ with $c$: $\tilde c(\tilde t)(U^1)^2T_t=c(t)$. Then the coefficient of $u_x$ results in $X^1_t=X^0_t=0$, which means $X^1=\delta_1$ and $X^0=\delta_2$ are arbitrary constants with $\delta_1\neq0$. The remaining equation, resulting from the coefficient of $u$, gives the relation $\tilde b$ and $b$, namely, $\tilde b(\tilde t)U^1T_t=U^1 b(t)+U^1_t$.
Providing the notations $U^1=\varphi(t)$, and $T=\theta(t)$ we get the statement of the theorem.
\end{proof}

Using Theorem~\ref{theorem:equivalence-group}, we can find the conditions for arbitrary elements $a$, $b$, and $c$, for which variable coefficient Newell--Whitehead--Segel equations are reducible to constant coefficient equations from the same class by point transformations. To derive such a condition we set
 $\tilde a$, $\tilde b$ and $\tilde c$ to be constants in the formulas~\eqref{eq_theor} and find compatibility condition for the obtained system. This results in the statement.

\begin{theorem}\label{theorem2}
A variable-coefficient equation from class~\eqref{eq_fish} is reduced to a constant-coefficient equation from the same class by a point transformation
if and only if for some constant~$\lambda$
the corresponding coefficients $a(t)$, $b(t)$ and $c(t)$ satisfy the condition
\begin{equation}\label{criterion1}
\frac b{a^2}+\frac12\frac{\big(c/a^2\big)_t}{c}=\lambda.
\end{equation}
\end{theorem}
The criterion~\eqref{criterion1} is rather useful for checking whether a given Newell--Whithead--Segel equation with time-dependent coefficients is similar to a constant coefficient equation from the same class.
In~\cite{triki} ``solutions'' were found for equations~\eqref{eq_fish} with
 $b(t)=c_1k^2a^2(t)$ and $c(t)=c_2k^2a^2(t)$, where $c_1$, $c_2$ and $k$ are constants. It is easy to see that for such
 values of $b(t)$ and $c(t)$ the condition~\eqref{criterion1} is satisfied.
 In Section~\ref{section5} we show how to get wide families of non-stationary solutions for the subclass of equations, whose coefficients satisfy~\eqref{criterion1} using the equivalence method.

Equivalence transformations allow us to simplify the initial class essentially. The arbitrary element $b(t)$ can be set to zero whereas $a(t)$ to a nonzero constant, for example, to one. Indeed, the transformation
\begin{equation}\label{tr}
\textstyle \tilde t =\int a^2(t){\rm d}t,\quad\tilde x=x,\quad\tilde u={\rm e}^{-\int b(t){\rm d} t}u
\end{equation}
maps class~\eqref{eq_fish} to its subclass
\begin{equation}\label{eq_NWS}
u_t=u_{xx}-c(t)u^3.
\end{equation}
The tildes in the latter equation are omitted.

Admissible transformations of the class~\eqref{eq_NWS} can be easily derived from Theorem~\ref{theorem:equivalence-group}, where we set $\tilde a^2=a^2=1$ and $\tilde b=b=0$. It guarantees the complete result since superclass~\eqref{eq_fish} of class~\eqref{eq_NWS}, is normalized. The result is summarized in the following statement.
\begin{theorem}
Class~\eqref{eq_NWS} is normalized. The equivalence groupoid of~\eqref{eq_NWS} is generated by transformations which form its usual equivalence group $G^\sim_1\colon$
\[
\tilde t=\delta_1{}^2t+\delta_0,\quad \tilde x=\delta_1x+\delta_2, \quad
\tilde u=\delta_3u, \quad
\tilde c(\tilde t)=\frac1{\delta_1{}^2\delta_3{}^2}c(t),\quad
\]
where $\delta_i$, $i=0,1,2,3$, are arbitrary constants with $\delta_1\delta_3\not=0$.
\end{theorem}
Therefore, we reduce the problem of classification of reduction operators for class~\eqref{eq_fish} up to the $G^\sim$-equivalence to the similar problem for class~\eqref{eq_NWS}, that contains only one arbitrary element~$c(t)$, up to the $G^\sim_1$-equivalence.

\section{Lie symmetries}\label{section3}

The group classification problem for the class \eqref{eq_NWS} is performed using the standard techni\-que~\smash{\cite{Olver1986,Ovsiannikov1982}}. We
search for operators of the form
\begin{equation}\label{eq_X}
X=\tau(t,x,u)\partial_t+\xi(t,x,u)\partial_x+\eta(t,x,u)\partial_u,
\end{equation}
which generate one-parameter Lie groups of point symmetry transformations for equations from class \eqref{eq_NWS}. Here we require that
\[
X^{(2)}\left \{u_t-u_{xx}+c(t)u^3\right \}=0
\]
modulo equations \eqref{eq_NWS}, where $X^{(2)}$ is the second prolongation of the operator $X$~\cite{Olver1986}. Firstly, we note that, since the class \eqref{eq_NWS} is an evolution equation which is
a polynomial in the pure derivatives of $u$ with respect to $x$, it can be shown that $\tau=\tau(t)$ and $\xi=\xi(t,x)$ \cite{Kingston&Sophocleous1998}. Using these simplifications, after elimination of $u_t$, the above equation takes the form
\[
(2\xi_x-\tau_t)u_{xx}-\eta_{uu}u_x^2+(\xi_{xx}-\xi_t-2\eta_{xu})u_x+(c_t\tau+c\tau_t-c\eta_u)u^3+3c\eta u^2+\eta_t-\eta_{xx} =0.
\]
The coefficients of the derivatives of $u$ with respect to $x$ in this identity provide a system of determining equations that enable us to derive the functional forms of the coefficient functions $\tau$,~$\xi$ and $\eta$ and also that of the arbitrary element $c(t)$. From the coefficients of $u_{xx}$, $u_x^2$ and $u_x$ we find that
\[
\xi=\tfrac 12\tau_t x+\psi(t),\quad \eta=-\left(\tfrac 18\tau_{tt} x^2+\tfrac 12\psi_t x -\varphi(t)\right)u+\zeta(t,x).
\]
Using these results, the term independent of derivatives implies that $\zeta=0$, $\tau=\kappa_1 t+\kappa_2$, $\psi=\kappa_3$, $\varphi=\kappa_4$ and $c(t)$ satisfies the relation
\begin{equation}\label{eq_class}
(\kappa_1t+\kappa_2)c_t+(\kappa_1+2\kappa_4)c=0,
\end{equation}
where $\kappa_i$, $i=1,2,3,4$, are arbitrary constants. Therefore, the infinitesimal generators~\eqref{eq_X} have the general form
\[
X=(\kappa_1t+\kappa_2)\partial_t+\left(\tfrac12\kappa_1x+\kappa_3\right)\partial_x+\kappa_4u\partial_u.
\]
Equation~\eqref{eq_class} is the classifying equation that appears during solving the group classification problems rather frequently, see, for example,~\cite{Popovych&Ivanova2004NVCDCEs,VSL}. If $c(t)$ is an arbitrary function then the classifying equation gives $\kappa_1=\kappa_2=\kappa_4=0$.
The corresponding basis of the maximal Lie symmetry algebra~$A^{\max}$ is the one-dimensional algebra $\langle\partial_x\rangle$.
The extensions of the maximal Lie symmetry algebra are possible if and only if
 $c(t)$ is power or exponential function of $t$, or a constant.
 Therefore, $c(t)$
 either takes the form $c(t)=\varepsilon t^{\rho}$ or $c(t)=\varepsilon {\rm e}^{\pm t}$, where $\rho$ is an arbitrary constant, and $\varepsilon=\pm1$ modulo the equivalence transformations~\eqref{eq_theor}. Substituting these forms of $c(t)$ into equation~\eqref{eq_class} we get the corresponding values of $\kappa_i$, $i=1,2,3,4$, and write down the bases of $A^{\max}$.
 The results are summarised in the following theorem.
\begin{theorem}
The kernel of maximal Lie symmetry algebras of equations from the class~\eqref{eq_NWS} is the one-dimensional algebra $\langle\partial_x\rangle$.
A complete list of $G^\sim$-inequivalent Lie symmetry extensions in class~\eqref{eq_NWS} is exhausted by the cases {\rm 1--3} given in Table~{\rm \ref{TableLieSym1}}.
\end{theorem}
\begin{table}[t!]
\begin{center}
\renewcommand{\arraystretch}{1.3}
\caption{\label{TableLieSym1}The group classification of class~\eqref{eq_NWS} up to the $G^\sim_1$-equivalence.}

\medskip
\begin{tabular}{ccccl}
\hline
no.&$c(t)$&\hfil Basis of $A^{\max}$ \\
\hline
0&$\forall$&$\partial_x$
\\
1&$\varepsilon t^\rho$&$\partial_x,\ 2t\partial_t+x\partial_x-(\rho+1)u\partial_u$
\\
2&$\varepsilon {\rm e}^{\pm t}$&$\partial_x,\ 2\partial_t\mp u\partial_u$
\\
3&$\varepsilon$&$\partial_x,\ \partial_t,\ 2t\partial_t+x\partial_x-u\partial_u$
\\
\hline
\end{tabular}
\\[1.5ex]
\parbox{95mm}{Here $\rho$ is an arbitrary nonzero constant, $\varepsilon=\pm1\bmod G^\sim_1$.}
\end{center}\end{table}

Table~\ref{TableLieSym1} represents also the group classification results for class~\eqref{eq_fish} up to the $G^\sim$-equivalence. We recall that
$a^2(t)=1\bmod G^\sim$, $b(t)=0\bmod G^\sim$ for all the cases of Lie symmetry extension.

For the practical use of the group classification results it is convenient to have also the list of Lie symmetry extensions which is not simplified by equivalence transformations. To get such a~list we use the algorithm described in~\cite{Vaneeva2012}.
Firstly we write down the most general forms of the function~$c(t)$ that correspond to equations from class~\eqref{eq_NWS} with Lie symmetry extensions. These are the cases:
\begin{enumerate}\samepage
\item[1)] $c(t)=\mu (\gamma t+\delta)^\rho$:\quad
$A^{\max}=\left\langle\partial_x,\,2(\gamma t+\delta)\partial_t+\gamma x\partial_x-\gamma(\rho+1)u\partial_u\right\rangle$;

\item[2)] $c(t)=\mu {\rm e}^{\sigma t}$:\quad
$A^{\max}=\langle\partial_x,\,2\partial_t-\sigma u\partial_u\rangle$;

\item[3)] $c(t)=\mu$:\quad
$A^{\max}=\langle\partial_x,\,\partial_t,\,2t\partial_t+x\partial_x-u\partial_u\rangle$.
\end{enumerate}
Here $\mu$, $\gamma,$ $\rho$ and $\sigma$ are arbitrary nonzero constants and $\delta$ is an arbitrary constant.

Using the transformation~\eqref{tr} and the latter classification list it's easy to obtain the classification list for class~\eqref{eq_fish} where arbitrary elements are not gauged by the equivalence transformations.
The results are summarized in Table~\ref{TableLieSym2}.
\begin{table}[t!]
\begin{center}
\renewcommand{\arraystretch}{1.3}
\caption{\label{TableLieSym2}The group classification of class~\eqref{eq_fish}
without usage of the equivalence group.}

\medskip
\begin{tabular}{ccccl}
\hline
no.&$c(t)$&\hfil Basis of $A^{\max}$ \\[1ex]
\hline
0&$\forall$&$\partial_x$
\\
1&$\mu {a^2}{\rm e}^{-2\int b\,{\rm d}t}(\gamma T+\delta)^\rho$&$\partial_x,\ \ \frac2{{a^2}}(\gamma T+\delta)\partial_t+\gamma x\partial_x+\bigl(\frac2{{a^2}}(\gamma T+\delta)b-(\rho+1)\gamma\bigr)u\partial_u$
\\
2&$\mu {a^2}{\rm e}^{\sigma T -2\int b\,{\rm d}t}$&$\partial_x,\ \ \frac2{{a^2}}\partial_t+\bigl(\frac{2b}{{a^2}}-\sigma\bigr)u\partial_u$
\\
3&$\mu {a^2}{\rm e}^{-2\int b\,{\rm d}t}$&$\partial_x,\ \ \frac1{{a^2}}\left(\partial_t+bu\partial_u\right),\ \ \frac2{{a^2}}\partial_t+x\partial_x+\bigl(\frac{2b}{{a^2}}-1\bigr)u\partial_u$
\\
\hline
\end{tabular}
\\[1.5ex]
\parbox{155mm}{Here $a=a(t)$ and $b=b(t)$ are arbitrary nonvanishing smooth functions, $T=\int\!a^2(t)\,{\rm d}t$;
$\mu$, $\sigma$, $\delta$ and~$\rho$ are arbitrary constants with $\mu\sigma\rho\neq0$.}
\end{center}
\end{table}
The latter list reveals the Newell--Whitehead--Segel equations which are of more interest for applications and for which the classical Lie reduction method can be utilized. It is also necessary for the study of nonclassical reduction operators that we perform in the next section to get truly nontrivial ones, i.e., those which are not equivalent to Lie reduction operators.

\section{Nonclassical method}\label{section4}
Given a $(1{+}1)$-dimensional evolution equation with the independent variables $t$ and $x$ and the dependent variable $u$, its reduction operators have the general form~\eqref{eq_X} with $(\tau,\xi)\not =(0,0)$.
The reduction operators~\eqref{eq_X} with nonvanishing
coefficients of $\partial_t$ are regular, and the other its reduction operators are singular~\cite{KunzingerPopovych}; see
also~\cite{Boyko}. The singular case $\tau=0$ was exhaustively investigated for general evolution equation in~\cite{KunzingerPopovych,Zhdanov&Lahno1998}.

Consider the case $\tau\ne0$.
We can assume $\tau=1$ up to the usual equivalence of reduction operators. This equivalence relation means that reduction operators $X$ and $\tilde X$ are equivalent if $\tilde X=\Lambda(t,x,u)X$,
where $\Lambda(t,x,u)$ is a nonvanishing smooth function of its arguments.
Then the nonclassical invariance criterion implies the following determining equations for the coefficients $\xi$ and $\eta$, and also for the arbitrary element $c(t)$:
\begin{gather}\label{EqDetForRedOps}\arraycolsep=0ex
\begin{array}{l}
\xi_{uu}=0,\quad
\eta_{uu}=2(\xi_{xu}-\xi\xi_u),\\[0.5ex]
\eta_{t}-\eta_{xx}+2\xi_{x}\eta+\left(2\xi_{x}-\eta_{u}\right)cu^3+3\eta cu^2+c_tu^3=0,\\[0.5ex]
\xi_{t}-\xi_{xx}+2\xi\xi_{x}-2\xi_{u}\eta+2\eta_{xu}-3\xi_{u}c u^3=0.
\end{array}
\end{gather}
Integration of the first two equations of system~\eqref{EqDetForRedOps} gives us the following
expressions for the coefficients $\xi$ and $\eta$
\begin{gather*}
\xi=fu+g,\quad
\eta=-\tfrac13f^2u^3+(f_x-fg)u^2+hu+k,
\end{gather*}
where $f=f(t,x)$, $g=g(t,x)$, $h=h(t,x)$ and $k=k(t,x)$.
We further substitute the derived forms of $\xi$ and $\eta$ into the rest two equations of system~\eqref{EqDetForRedOps} and split the resulting equations with respect to variable $u$.
This leads to a system of nine determining equations involving operator coefficients $f$, $g$, $h$, and $k$ as well as the arbitrary element $c(t)$ of class~\eqref{EqDetForRedOps}. One of the equations is $f(9c-2f^2)=0$. The further consideration splits into two cases $f\neq0$ and $f=0$.

{\bf I.} If $f\neq0$, then $9c-2f^2=0$, which means $f_x=0$ and $f$ is a function of $t$ only. Then the rest of the determining equations imply $g=k=0$, $h=\alpha$, $f=\beta {\rm e}^{2\alpha t}$, and $c=\frac29\beta^2{\rm e}^{4\alpha t}$, where $\alpha$ and $\beta\neq0$ are constants.
Therefore, the equation
\begin{equation}\label{eq_red_op_1}
u_t=u_{xx}-\frac29\beta^2{\rm e}^{4\alpha t}u^3
\end{equation}
admits the nonclassical reduction operator
\[
X_1=\partial_t+\beta {\rm e}^{2\alpha t}u\partial_x+\left(\alpha-\frac13\beta^2{\rm e}^{4\alpha t}u^2\right)u\partial_u.
\]
The constants $\alpha$ and $\beta$ can be additionally gauged by equivalence transformations, see case 1 of Table~~\ref{TableNonclassicalSym}.

{\bf II.} If $f=0$, then $k=0$, $h=-g_x-\frac12\frac{\dot c}c$ and the rest of the determining equations are
\begin{gather*}\label{deteqs2}
g_t+2gg_x-3g_{xx}=0,\\
g_{tx}+2g_x^2-g_{xxx}+\dfrac{\dot c}c g_x+\dfrac12\frac{\rm d}{{\rm d}t}\left(\dfrac{\dot c}c\right)=0.
\end{gather*}
This system of two partial differential equations for the function $g(t,x)$, one of which involves arbitrary element~$c(t)$ of the class.
The investigation of compatibility of this system implies that~$c(t)$ can be only a power, exponential or constant function, otherwise the system is inconsistent.
Truly non-Lie reduction operators arise only if~$c(t)$ is either an exponential function or a constant.
The list of the equations admitting nontrivial nonclassical reduction operators with $\xi_u=0$ is the following:
\begin{gather}\nonumber 
u_t=u_{xx}-\mu u^3\colon
\\
X_2=\partial_t-\frac3{x}\partial_x-\frac3{x^2}u\partial_u.\nonumber\\[0.5ex]
\label{eq_red_op_3}
u_t=u_{xx}-\mu {\rm e}^{\sigma t} u^3\colon
\\
X_3=\partial_t-\frac32\sqrt{\sigma}\tanh\left(\frac{\sqrt{\sigma}}2x\right)\partial_x
-\frac3{4}\sigma\left(\tanh^2\left(\frac{\sqrt{\sigma}}2x\right)-\frac13\right)u\partial_u,\ \ \sigma>0;\nonumber\\
X_4=\partial_t-\frac32\sqrt{\sigma}\coth\left(\frac{\sqrt{\sigma}}2x\right)\partial_x
-\frac3{4}\sigma\left(\coth^2\left(\frac{\sqrt{\sigma}}2x\right)-\frac13\right)u\partial_u,\ \ \sigma>0;\nonumber\\
X_5=\partial_t+\frac32\sqrt{-\sigma}\tan\left(\frac{\sqrt{-\sigma}}2x\right)\partial_x
+\frac3{4}\sigma\left(\tan^2\left(\frac{\sqrt{-\sigma}}2x\right)+\frac13\right)u\partial_u,\ \ \sigma<0.\nonumber
\end{gather}
Here $\mu$ and $\sigma$ are arbitrary nonzero constants. Both of them can be gauged by the equivalence transformations to be equal to $1$ or $-1$ depending on their signs, namely $\mu\mapsto\sign\mu$, and $\sigma\mapsto\sign\sigma$.

We summarize the results on classification of nonclassical reduction operators of equations~\eqref{eq_NWS} up to the $G^\sim_1$-equivalence in Table~\ref{TableNonclassicalSym}. In all the cases of Table~\ref{TableNonclassicalSym} $\varepsilon=\pm1$. The same table represents the results on classification of nonclassical reduction operators of equations~\eqref{eq_fish} up to the $G^\sim$-equivalence ($a(t)=1\bmod G^\sim$ and $b(t)=0\bmod G^\sim$ in this case).
\begin{table}[t!]
\begin{center}
\renewcommand{\arraystretch}{1.7}
\caption{\label{TableNonclassicalSym}Nonclassical
reduction operators of equations~\eqref{eq_NWS}.}

\medskip
\begin{tabular}{ccl}
\hline
no.&$c(t)$&\hfil Reduction operators \\
\hline
1&${\rm e}^{\pm t}$&$\partial_t+\frac{3\sqrt{2}}2 {\rm e}^{\pm\frac12 t}u\partial_x-\frac12\left(3{\rm e}^{\pm t}u^2\mp\frac12\right)u\partial_u$\\
&$\varepsilon {\rm e}^{t}$&$\partial_t-\frac32\tanh\left(\frac{1}2x\right)\partial_x
-\frac3{4}\left(\tanh^2\left(\frac12x\right)-\frac13\right)u\partial_u$\\
&&$\partial_t-\frac32\coth\left(\frac12x\right)\partial_x
-\frac3{4}\left(\coth^2\left(\frac12x\right)-\frac13\right)u\partial_u$
\\
&$\varepsilon {\rm e}^{-t}$&$\partial_t+\frac32\tan\left(\frac12x\right)\partial_x
-\frac3{4}\left(\tan^2\left(\frac12x\right)+\frac13\right)u\partial_u$\\
\hline
2&$\varepsilon$&$\partial_t-\frac3{x}\partial_x-\frac3{x^2}u\partial_u$\\
\hline
\end{tabular}
\end{center}
\end{table}

Theorem~\ref{theorem2} implies that equations~\eqref{eq_red_op_1} and~\eqref{eq_red_op_3} are reducible to constant coefficient Newell--Whitehead--Segel equations~\eqref{eq_fish} by equivalence transformations from the group $G^\sim$.
Indeed, the transformation
\[
\tilde t=t, \quad \tilde x=x, \quad \tilde u={\rm e}^{\frac\sigma2t}u
\]
maps equation~\eqref{eq_red_op_3} to the equation
\[
\tilde u_{\tilde t}=\tilde u_{\tilde x\tilde x}+\frac\sigma 2 \tilde u-\mu{\tilde u}^3.
\]
The latter observation means that direct reduction of equations~\eqref{eq_red_op_1} and~\eqref{eq_red_op_3} using the nonclassical symmetry operators is not the optimal way for finding their exact solutions.
More convenient way is the reduction of their constant coefficient counterparts (or immediate usage of exact solutions of constant coefficient equations, if such solutions are known) and then derivation of exact solutions by the equivalence method, see the related discussion in~\cite{Popovych&Vaneeva2010}. The next section is devoted to construction of exact solutions for equations from class~\eqref{eq_fish}
using the equivalence transformations.

\section{Exact solutions}\label{section5}

Theorem~\ref{theorem2} implies that equations of the form
\begin{gather}\label{eq_fish1}
u_t=a^2(t)u_{xx}+\left(\lambda a^2(t)+\frac{\dot a(t)}{a(t)}-\frac12\frac{\dot c(t)}{c(t)}\right)u-c(t)u^3,
\end{gather}
where $a(t)$ and $c(t)$ are nonvanishing smooth functions and $\lambda$ is a nonzero constant,
are similar to the constant-coefficient equation
 \begin{equation}\label{cubic}
u_t=u_{xx}+\varepsilon u- u^3
\end{equation}
 with $\varepsilon=\sign\lambda$. The latter equation is well studied by various techniques and a~number of its exact solutions are known, see, e.g., \cite[p.~177]{Polyanin&Zaitsev2012} and \cite{VPS2009}, and references therein. The similarity is established by the transformation
\begin{equation}\label{trr}
\tilde t=|\lambda|\int a^2(t){\rm d}t,\quad \tilde x=\sqrt{|\lambda|}\,x,\quad \tilde u=\frac1{a(t)}\sqrt{\frac{c(t)}{|\lambda|}}\,u.
\end{equation} for the case $\lambda\neq0$ and by the transformation
\begin{equation}\label{trr1}
\tilde t=\int a^2(t){\rm d}t,\quad \tilde x=x, \quad \tilde u=\frac{\sqrt{c(t)}}{a(t)}\,u,
\end{equation}
otherwise.
There are obvious restrictions for this transformations to connect two real valued exact solutions for the physical case $t>0$.
It works fine for all functions $c(t)>0$ when $t>0$, for example for power coefficient $c(t)$ that is used most frequently in applications.

We illustrate the possibility of generation of solutions for equations~\eqref{eq_fish1} by the following example.
The transformation~\eqref{trr}
maps the known traveling wave solution
\begin{equation*}
u=\frac12-\frac12\tanh\left(\frac{\sqrt{2}}{4}
x-\frac{3}{4}t\right)
\end{equation*}
 of the constant-coefficient equation~\eqref{cubic} with $\varepsilon=1$~\cite{Wang1988} to new exact solution{\samepage
\[
u=\frac12{ a(t)}\sqrt{\dfrac{\lambda}{c(t)}}\left(1-\tanh\biggl(\frac{\sqrt{2\lambda}}4x-\frac34\lambda\int a^2(t){\rm d}t\biggr)\right)
\]
of variable-coefficient equation~\eqref{eq_fish1} with $\lambda>0$ and $c(t)>0$ for $t>0$.}

A number of other exact solutions of the equation~\eqref{cubic} are collected in~\cite{NikitinBarannyk2004,Polyanin&Zaitsev2012, VPS2009}.
We consider the exact solutions of the equation~\eqref{cubic} collected in~\cite{VPS2009} and apply to them either transformation~\eqref{trr} in the case $\lambda\neq0$ or transformation~\eqref{trr1}, otherwise. As a result we obtain wide families of exact solutions of variable coefficient Newell--Whitehead--Segel equations~\eqref{eq_fish1}.

Hereafter $T=|\lambda|\int a^2(t){\rm d}t;$ the functions ${\rm cn}(z,k)$, ${\rm sn}(z,k)$, and ${\rm ds}(z,k)$
are Jacobian elliptic functions \cite{WhittakerWatson}.

$\lambda>0\colon$
\begin{gather*}
u= { a(t)}\sqrt{\dfrac{\lambda}{c(t)}}\frac{C_1\exp\left(
\frac{\sqrt{2\lambda}}2x\right)-C_1'\exp\left(-\frac{\sqrt{2\lambda}}2x\right)}
{C_2\exp\left(-\frac32T\right)+C_1\exp\left(\frac{\sqrt{2\lambda}}2x\right)+C_1'\exp\left(-\frac{\sqrt{2\lambda}}2x\right)},
\\
u= { a(t)}\sqrt{\dfrac{\lambda}{c(t)}}C_1\exp\left(\tfrac32T\right)\sinh\left(\tfrac{\sqrt{2\lambda}}2x\right)
{\rm ds}\left(C_1\exp\left(\tfrac32T\right)\cosh\left(\tfrac{\sqrt{2\lambda}}2x\right)+C_2,\tfrac{\sqrt2}2\right),
\\
u= { a(t)}\sqrt{\dfrac{\lambda}{c(t)}}C_1\exp\left(\tfrac32T\right)\cosh\left(\tfrac{\sqrt{2\lambda}}2x\right)
{\rm ds}\left(C_1\exp\left(\tfrac32T\right)\sinh\left(\tfrac{\sqrt{2\lambda}}2x\right)+C_2,\tfrac{\sqrt2}2\right),
\\
u= { a(t)}\sqrt{\dfrac{\lambda}{c(t)}}\frac{C_1}2\exp\left(\tfrac32T\right)\sinh\left(\tfrac{\sqrt{2\lambda}}2x\right)\!\!
\frac{1+{\rm cn}\left(C_1\exp\left(\tfrac32T\right)\cosh\left(\tfrac{\sqrt{2\lambda}}2x\right)+C_2,\tfrac{\sqrt2}2\right)}{{\rm sn}
\left(C_1\exp\left(\tfrac32T\right)\cosh\left(\frac{\sqrt{2\lambda}}2x\right)+C_2,\tfrac{\sqrt2}2\right)},
\\
u= { a(t)}\sqrt{\dfrac{\lambda}{c(t)}}\frac{C_1}2\exp\left(\tfrac32T\right)\cosh\left(\tfrac{\sqrt{2\lambda}}2x\right)\!\!
\frac{1+{\rm cn}\left(C_1\exp\left(\tfrac32T\right)\sinh\left(\frac{\sqrt{2\lambda}}2x\right)+C_2,\tfrac{\sqrt2}2\right)}{{\rm sn}
\left(C_1\exp\left(\tfrac32T\right)\sinh\left(\frac{\sqrt{2\lambda}}2x\right)+C_2,\tfrac{\sqrt2}2\right)}.
\end{gather*}

$\lambda<0\colon$
\begin{gather*}
u= { a(t)}\sqrt{\dfrac{-\lambda}{c(t)}}\frac{\sin\left(\tfrac{\sqrt{-2\lambda}}2x\right)}{C_2\exp\left(\frac32T\right)+\cos\left(\tfrac{\sqrt{-2\lambda}}2x\right)},
\\
u= { a(t)}\sqrt{\dfrac{-\lambda}{c(t)}}C_1\exp\left(-\tfrac32T\right)\sin\left(\tfrac{\sqrt{-2\lambda}}2x\right){\rm ds}
\left(C_1\exp\left(-\tfrac32T\right)\cos\left(\tfrac{\sqrt{-2\lambda}}2x\right)+C_2,\tfrac{\sqrt2}2\right),
\\
u= { a(t)}\sqrt{\dfrac{-\lambda}{c(t)}}\frac{C_1}2\exp\left(-\tfrac32T\right)\cos\left(\tfrac{\sqrt{-2\lambda}}2x\right)\!\!
\frac{1+{\rm cn}\left(C_1\exp\left(-\tfrac32T\right)\sin\left(\tfrac{\sqrt{-2\lambda}}2x\right)+C_2,\tfrac{\sqrt2}2\right)}{{\rm sn}
\left(C_1\exp\left(-\tfrac32T\right)\sin\left(\tfrac{\sqrt{-2\lambda}}2x\right)+C_2,\tfrac{\sqrt2}2\right)}.
\end{gather*}

$\lambda=0\colon$
\begin{gather*}
u= 2\sqrt2\,x\,\frac{a(t)}{\sqrt{c(t)}}\,{\rm ds}\left(x^2+6\int a^2(t){\rm d}t,\tfrac{\sqrt2}2\right), \\
u= \sqrt2\,x\,\frac{a(t)}{\sqrt{c(t)}}\,
\dfrac{1+{\rm cn}\left(x^2+6\int a^2(t){\rm d}t,\tfrac{\sqrt2}2\right)}
{{\rm sn}\left(x^2+6\int a^2(t){\rm d}t,\tfrac{\sqrt2}2\right)},\\
u= \frac{a(t)}{\sqrt{c(t)}}\dfrac{2\sqrt2\, x}{{x}^2+6\int a^2(t){\rm d}t},\quad
u= \frac{a(t)}{\sqrt{c(t)}}\dfrac{\sqrt2}{x}, \\
u= \sqrt2\,\frac{a(t)}{\sqrt{c(t)}}\,{\rm ds}\left(x,\tfrac{\sqrt2}2\right), \quad
u= \dfrac{\sqrt2}2\,\frac{a(t)}{\sqrt{c(t)}}\,\dfrac{1+{\rm cn}\left(x,\tfrac{\sqrt2}2\right)}
{{\rm sn}\left(x,\tfrac{\sqrt2}2\right)}.
\end{gather*}

As equation~\eqref{eq_NWS} admits the equivalence transformation of the alternating sign $u\mapsto-u$ all the above presented solutions can also have the forms with the opposite sign.

\section{Conclusion}
We have performed an extended group analysis of variable coefficient Newell--Whitehead--Se\-gel~\eqref{eq_fish}. Along with the classical group classification problem a more difficult classification problem of nonclassical reduction operators ($Q$-conditional symmetries) has been solved. We have also derived the equivalence groupoid and the criterion of reducibility of variable coefficient Newell--Whitehead--Segel equations to constant coefficient equations from the same class. The latter results form a~basis for the usage of the equivalence method for finding exact solutions. As a~result wide families of exact solutions for variable coefficient Newell--Whitehead--Segel~\eqref{eq_fish} have been constructed.

\subsection*{Acknowledgments} The authors are grateful to Roman Popovych for useful remarks. OV acknowledges the financial support provided by the NAS of Ukraine under
the project 0118U003803. 
The first three authors are  partially   supported by the NAS of Ukraine under the project 0116U003059.


\begin{thebibliography}{99}
\footnotesize\itemsep=-2pt

\bibitem{ArrigoHill1995}
Arrigo D.J. and Hill J.M.
Nonclassical symmetries for nonlinear diffusion and absorption,
{\it Stud. Appl. Math.} {\bf 94} (1995), \href{https://doi.org/10.1002/sapm199594121}{21--39}.

\bibitem{ArrigoHillBroadbridge1993}
Arrigo D.J., Hill J.M. and Broadbridge P.
Nonclassical symmetry reductions of the linear diffusion equation with a nonlinear source,
{\it IMA~J. Appl. Math.} {\bf 52} (1994), \href{https://doi.org/10.1093/imamat/52.1.1}{1--24}.

\bibitem{Basarab2001}
 Basarab-Horwath P., Lahno V. and Zhdanov R. The structure of Lie algebras and the classification problem for partial differential equations, {\it Acta Appl. Math.} {\bf 69} (2001), \href{https://doi.org/10.1023/A:1012667617936}{43--94}, \href{https://arxiv.org/abs/math-ph/0005013}{arXiv:math-ph/0005013}.

\bibitem{Bluman1968}
Bluman G.W. {\it Construction of Solutions to Partial Differential Equations by the Use of Transformation Groups}, \href{http://resolver.caltech.edu/CaltechETD:etd-04082005-155822}{Ph.D.~Thesis}, California Institute of Technology, 1968.

\bibitem{Bluman&Cole1969}
 Bluman G.W. and Cole J.D.
The general similarity solution of the heat equation,
{\it J.~Math. Mech.} {\bf 18} (1969), \href{https://doi.org/10.1512/iumj.1969.18.18074}{1025--1042}.

\bibitem{Boyko}
Boyko V.M., Kunzinger M. and Popovych R.O. Singular reduction modules of differential equations, {\it J.~Math. Phys.} {\bf 57} (2016), \href{https://doi.org/10.1063/1.4965227}{101503, 34~pp.}; \href{https://arxiv.org/abs/1201.3223}{arXiv:1201.3223}.

\bibitem{Bradshaw2007}
Bradshaw-Hajek B H., Edwards M.P., Broadbridge P. and Williams G.H. Nonclassical symmetry solutions for reaction-diffusion equations
with explicit spatial dependence, {\it Nonlinear Anal.} {\bf 67} (2007),
\href{https://doi.org/10.1016/j.na.2006.09.022}{2541--2552}.

\bibitem{Cheviakov}
Cheviakov A.F. GeM software package for computation of symmetries and conservation laws of differential equations, {\it Comp. Phys. Comm.} {\bf 176} (2007), \href{https://doi.org/10.1016/j.cpc.2006.08.001}{48--61}.

\bibitem{Cheviakov2010}
Cheviakov A.F. Symbolic computation of local symmetries of nonlinear and linear partial and ordinary differential equations, {\it Math. Comput. Sci.} {\bf 4} (2010), \href{https://doi.org/10.1007/s11786-010-0051-4}{203--222}.

\bibitem{ClarksonKruskal1989}
 Clarkson P.A. and Kruskal M.D. New similarity solutions of the Boussinesq equation, {\it J.~Math. Phys.} {\bf 30}
(1989), \href{https://doi.org/10.1063/1.528613}{2201--2213}.

\bibitem{Clarkson&Mansfield1993}
 Clarkson P.A. and Mansfield E.L. Symmetry reductions and exact solutions of
a class of nonlinear heat equations,
{\it Phys.~D} {\bf 70} (1994), \href{https://doi.org/10.1016/0167-2789(94)90017-5}{250--288}; \href{https://arxiv.org/abs/solv-int/9306002}{arXiv:solv-int/9306002}.

\bibitem{Fisher1937} Fisher R.A. The wave of advance of advantageous genes,
{\it Ann. Eugenics} {\bf 7} (1937), \href{https://doi.org/10.1111/j.1469-1809.1937.tb02153.x}{353--369}.

\bibitem{Fushchych81}
Fushchych W.I. Symmetry in problems of mathematical physics, in Algebraic-Theortic Studies in
Mathematical Physics, Kiev, Inst. of Math. Ukrainian Acad. Sci., 1981, \href{https://www.imath.kiev.ua/~fushchych/papers/1981_1.pdf}{6--44} (in Russian).

\bibitem{Fushchych92}
Fushchych W.I. Conditional symmetry of equations
of nonlinear mathematical physics, in Symmetry Analysis of Equations of Mathematical Physics, Kyiv, Institute Mathematics, 1992, \href{https://www.imath.kiev.ua/~fushchych/papers/1992_3.pdf}{7--27}.


\bibitem{Fushchych95}
Fushchych W.I. Ansatz '95, {\it J.~Nonlinear Math. Phys.} {\bf 2} (1995), \href{https://doi.org/10.2991/jnmp.1995.2.3-4.2}{216--235}.

\bibitem{FushchichSerov1990}
 Fushchich W.I. and Serov N.I. Conditional invariance and reduction of nonlinear heat equation,
{\it Dokl. Akad. Nauk Ukrain. SSR Ser.~A} (1990), no.~7, \href{http://www.imath.kiev.ua/~fushchych/papers/1990_14.pdf}{24--27} (in Russian).

\bibitem{FushchichShtelenSerov-book}
Fushchich W.I., Shtelen W.M. and Serov N.I. {\it Symmetry Analysis and Exact Solutions of Equations of Nonlinear Mathematical Physics}, \href{https://doi.org/10.1007/978-94-017-3198-0}{Kluwer}, Dordrecht, 1993.

\bibitem{Fushchych&Shtelen&Serov&Popovych1992}
 Fushchich W.I., Shtelen W.M., Serov M.I. and Popovych R.O.
$Q$-conditional symmetry of the linear heat equation,
{\it Proc. Acad. of Sci. Ukraine} (1992), no.~12, \href{https://www.imath.kiev.ua/~fushchych/papers/1992_12.pdf}{28--33}.

\bibitem{Fushchych&Tsyfra1987}
 Fushchich W.I. and Tsifra I.M.
On a reduction and solutions of the nonlinear wave equations with broken symmetry,
{\it J.~Phys.~A: Math. Gen.} {\bf 20} (1987), \href{https://doi.org/10.1088/0305-4470/20/2/001}{L45--L48}.


\bibitem{GagnonWinternitz1993}
Gagnon L. and Winternitz P. Symmetry classes of variable coefficient nonlinear Schr\"odinger equations, {\it J.~Phys.~A: Math. Gen.} {\bf 26} (1993), \href{doi:10.1088/0305-4470/26/23/043}{7061--7076}.

\bibitem{GrundlandTafel1995}
Grundland A.M. and Tafel J., On the existence of nonclassical symmetries of partial differential equations,
{\it J.~Math. Phys.} {\bf 36} (1995), \href{https://doi.org/10.1063/1.531131}{1426--1434}.

\bibitem{GungorWinternitz2004}
G\"ung\"or F. and Winternitz P. Equivalence classes and symmetries of the variable coefficient Kadomtsev--Petviashvili equation, {\it Nonlinear Dynam.} {\bf 35} (2004), \href{https://doi.org/10.1023/B:NODY.0000027746.31782.6e}{381--396}

\bibitem{Hammond&Bortz2011} Hammond J.F. and Bortz D.M. Analytical solutions to Fisher's equation with time-variable coefficients, {\it Appl. Math. Comput.} {\bf 218} (2011),
\href{https://doi.org/10.1016/j.amc.2011.03.163}{2497--2508}.

\bibitem{Ivanova&Sophocleous2010}
Ivanova N.M. and Sophocleous C. On nonclassical symmetries
of generalized Huxley equations, \emph{Proceedings of the Fifth International Workshop ``Group Analysis of Differential Equations and Integrable Systems'' (Protaras, Cyprus, June 6--10, 2010)}, University of Cyprus, Nicosia, 2011, 91--98;
\href{https://arxiv.org/abs/1010.2388}{arXiv:1010.2388}.

\bibitem{Kingston1991}
Kingston J.G. On point transformations of evolution equations, {\it J.~Phys.~A: Math. Gen.} {\bf 24} (1991), \href{https://doi.org/10.1088/0305-4470/24/14/003/}{{\rm L}769--{\rm L}774}.

\bibitem{Kingston&Sophocleous1998}
Kingston J.G. and Sophocleous C.
On form-preserving point transformations of partial differential equations,
{\it J.~Phys.~A: Math. Gen.} {\bf 31} (1998), \href{https://doi.org/10.1088/0305-4470/31/6/010}{1597--1619}.

\bibitem{KunzingerPopovych}
Kunzinger M. and Popovych R.O. Singular reduction operators
in two dimensions, {\it J.~Phys.~A: Math. Theor.} {\bf 41} (2008), \href{https://doi.org/10.1088/1751-8113/41/50/505201}{505201, 24~pp.}; \href{https://arxiv.org/abs/0808.3577}{arXiv:0808.3577}.

\bibitem{KunzingerPopovych2009}
Kunzinger M. and Popovych R.O. Is a nonclassical symmetry a symmetry?, in {\it Proceedings of 4th Workshop ``Group Analysis of Differential Equations and Integrable Systems'' (Protaras, Cyprus, October 26--30, 2008)}, University of Cyprus, Nicosia, 2009, 107--120; \href{https://arxiv.org/abs/0903.0821}{arXiv:0903.0821}.

\bibitem{levi1989}
Levi D. and Winternitz P. Nonclassical symmetry reduction: example of the Boussinesq equation, {\it J.~Phys.~A: Math. Gen.} {\bf 22} (1989), \href{https://doi.org/10.1088/0305-4470/22/15/010}{2915--2924}.

\bibitem{murray2002}
Murray J.D. {\it Mathematical Biology I: An Introduction}, 3rd~ed.,
\href{https://doi.org/10.1007/b98868}{Springer}, New York, 2002.

\bibitem{NewellWhitehead} Newell A.C. and Whitehead J.A. Finite bandwidth, finite amplitude convection, {\it J.~Fluid Mech.} {\bf 38} (1969), \href{https://doi.org/10.1017/S0022112069000176}{279--303}.

\bibitem{NikitinBarannyk2004}
Nikitin A.G. and Barannyk T.A.
Solitary wave and other solutions for nonlinear heat equations,
{\it Cent. Eur.~J. Math.} {\bf 2} (2004), \href{https://doi.org/10.2478/BF02475981}{840--858}; \href{https://arxiv.org/abs/math-ph/0303004}{arXiv:math-ph/0303004}.

\bibitem{Nucci}
Nucci M.C. and Clarkson P.A. The nonclassical method is more general than the direct method for symmetry reductions. An example of the Fitzhugh--Nagumo equation, {\it Phys. Lett.~A} {\bf 164} (1992),
\href{https://doi.org/10.1016/0375-9601(92)90904-Z}{49--56}.

\bibitem{OgunKart2007}
\"{O}\u{g}\"{u}n A. and Kart C. Exact solutions of Fisher and generalized Fisher
equations with variable coefficients, {\it Acta Math. Sin. (Engl. Ser.)} {\bf 23} (2007), \href{https://doi.org/10.1007/s10255-007-0395}{563--568}.

\bibitem{Olver1986}
Olver P. {\it Applications of Lie Groups to Differential Equations},
\href{https://doi.org/10.1007/978-1-4684-0274-2}{Springer-Verlag}, New-York, 1986.

\bibitem{OlverRosenau1987}Olver P.J. and Rosenau P. Group-invariant solutions of differential equations, {\it SIAM J. Appl. Math.} {\bf 47} (1987), \href{https://doi.org/10.1137/0147018}{263--278}.

\bibitem{OlverVorobev1996}
Olver P.J. and Vorob’ev E.M. Nonclassical and conditional symmetries, in {\it CRC Handbook of Lie Group
Analysis of Differential Equations}, Vol.~3, Editor N.H.~Ibragimov, CRC Press, Boca Raton, Florida, 1996,
291--328.

\bibitem{OpanasdenkoBihloPopovych2017}
Opanasenko S., Bihlo A. and Popovych R.O. Group analysis of general Burgers--Korteweg--de~Vries equations, {\it J.~Math. Phys.} {\bf 58} (2017), \href{https://doi.org/10.1063/1.4997574}{081511, 37~pp.}; \href{https://arxiv.org/abs/1703.06932}{arXiv:1703.06932}.

\bibitem{Ovsiannikov1982}
Ovsiannikov L.V. {\it Group Analysis of Differential Equations},
Academic Press, New York, 1982.

\bibitem{PocheketaPopovych2017}
Pocheketa O.A. and Popovych R.O. Extended symmetry analysis of generalized Burgers equations, {\it J.~Math. Phys.} {\bf 58} (2017), \href{https://doi.org/10.1063/1.5004134}{101501, 28~pp.}; \href{https://arxiv.org/abs/1603.09377}{arXiv:1603.09377}.

\bibitem{Polyanin&Zaitsev2012}
Polyanin A.D. and Zaitsev V.F. {\it Handbook of Nonlinear Partial Differential Equations}, 2nd~ed.,
Chapman \& Hall / CRC, Boca Raton, 2012; online version is available at \url{http://eqworld.ipmnet.ru/}.

\bibitem{Popovych2006}
Popovych R.O.
Classification of admissible transformations of differential equations,
\emph{Collection of Works of Institute of Mathematics} {\bf 3} (2006), no.~2, \href{https://www.imath.kiev.ua/~appmath/Collections/collection2006.pdf}{239--254}.

\bibitem{Popovych2006b}
Popovych R.O. No-go theorem on reduction operators of linear second-order parabolic equations, \emph{Collection of Works of Institute of Mathematics} {\bf 3} (2006), no.~2, \href{https://www.imath.kiev.ua/~appmath/Collections/collection2006.pdf}{231--238}.

\bibitem{Popovych2008}
Popovych R.O. Reduction operators of linear second-order parabolic equations, {\it J.~Phys.~A: Math. Theor.} {\bf 41} (2008), \href{https://doi.org/10.1088/1751-8113/41/18/185202}{185202, 31~pp.}; \href{https://arxiv.org/abs/0712.2764}{arXiv:0712.2764}.

\bibitem{PopovychBihlo}
Popovych R.O. and Bihlo A. Symmetry preserving parameterization schemes, {\it J.~Math. Phys.} {\bf 53} (2012),
\href{https://doi.org/10.1063/1.4734344}{073102, 36~pp.}; \href{https://arxiv.org/abs/1010.30105}{arXiv:1010.3010}.

\bibitem{Popovych&Ivanova2004NVCDCEs}
Popovych R.O. and Ivanova N.M. New results on group classification of
nonlinear diffusion--convection equations, {\it J.~Phys.~A: Math. Gen.} {\bf 37}
(2004), \href{https://doi.org/10.1088/0305-4470/37/30/011}{7547--7565}; \href{https://arxiv.org/abs/math-ph/0306035}{arXiv:math-ph/0306035}.

\bibitem{Popovych&Kunzinger&Eshraghi2010}
Popovych R.O., Kunzinger M. and Eshraghi H.
Admissible transformations and normalized classes of nonlinear Schr\"odinger equations,
{\it Acta Appl. Math.} {\bf 109} (2010), \href{https://doi.org/10.1007/s10440-008-9321-4}{315--359}; \href{https://arxiv.org/abs/math-ph/0611061}{arXiv:math-ph/0611061}.

\bibitem{Popovych&Vaneeva2010}
Popovych R.O. and Vaneeva O.O. More common errors in finding exact solutions of nonlinear differential equations:
Part~I, {\it Commun. Nonlinear Sci. Numer. Simul.} {\bf 15} (2010), \href{https://doi.org/10.1016/j.cnsns.2010.01.037}{3887--3899}; \href{https://arxiv.org/abs/0911.1848}{arXiv:0911.1848}.

\bibitem{Popovych&VaneevaIvaanova2007}
Popovych R.O., Vaneeva O.O. and Ivanova N.M. Potential nonclassical symmetries and solutions of fast diffusion equation, {\it Phys. Lett.~A} {\bf 362} (2007), \href{https://doi.org/10.1016/j.physleta.2006.10.015}{166--173}; \href{https://arxiv.org/abs/math-ph/0506067}{arXiv:math-ph/0506067}.

\bibitem{Segel}
Segel L.A. Distant side-walls cause slow amplitude modulation of cellular convection, {\it J.~Fluid Mech.} {\bf38} (1969),
\href{https://doi.org/10.1017/S0022112069000127}{203--224}.

\bibitem{triki} Triki H. and Wazwaz A.M. Trial equation method for solving the generalized Fisher equation with variable coefficients, {\it Phys. Lett.~A} {\bf 380} (2016), \href{https://doi.org/10.1016/j.physleta.2016.02.002}{1260--1262}.

 \bibitem{Vaneeva2012}
Vaneeva O.O.
Lie symmetries and exact solutions of variable coefficient mKdV equations: an equivalence based approach,
{\it Commun. Nonlinear Sci. Numer. Simul.} {\bf 17} (2012), \href{https://doi.org/10.1016/j.cnsns.2011.06.038}{611--618}; \href{https://arxiv.org/abs/1104.1981}{arXiv:1104.1981}.

\bibitem{VJPS2007}
Vaneeva O.O., Johnpillai A.G., Popovych R.O. and Sophocleous C.
Enhanced group analysis and conservation laws of variable coefficient reaction-diffusion equations with power nonlinearities,
{\it J. Math. Anal. Appl.} {\bf 330} (2007), \href{https://doi.org/10.1016/j.jmaa.2006.08.056}{1363--1386}; \href{https://arxiv.org/abs/math-ph/0605081}{arXiv:math-ph/0605081}.

\bibitem{VPS2009}
Vaneeva O.O., Popovych R.O. and Sophocleous C.
Enhanced group analysis and exact solutions of variable coefficient semilinear diffusion equations with a power source,
{\it Acta Appl. Math.} {\bf 106} (2009), \href{https://doi.org/10.1007/s10440-008-9280-9}{1--46}; \href{https://arxiv.org/abs/0708.3457}{arXiv:0708.3457}.

\bibitem{VPS2013}
Vaneeva O.O., Popovych R.O. and Sophocleous C.
Group classification of the Fisher equation with time-dependent coefficients,
\emph{Proceedings of the Sixth International Workshop ``Group Analysis of Differential Equations and Integrable Systems'' (Protaras, Cyprus, June 17--21, 2012)}, University of Cyprus, Nicosia, 2013, \href{http://www.mas.ucy.ac.cy/~symmetry/GADEISVI.pdf}{225--237}.

\bibitem{VPS2014}
Vaneeva O.O., Popovych R.O. and Sophocleous C.
Equivalence transformations in the study of integrability,
{\it Phys. Scr.} {\bf 89} (2014), \href{https://doi.org/10.1088/0031-8949/89/03/038003}{038003, 9~pp.}; \href{https://arxiv.org/abs/1308.5126}{arXiv:1308.5126}.

\bibitem{VaneevaPosta2017}
Vaneeva O. and Po\v{s}ta S. Equivalence groupoid of a class of variable coefficient Korteweg--de Vries equations, {\it J.~Math. Phys.} {\bf 58} (2017), \href{https://doi.org/10.1063/1.5004973}{101504, 12~pp.}; \href{https://arxiv.org/abs/1604.06880}{arXiv:1604.06880}.

\bibitem{VSL}
Vaneeva O.O., Sophocleous C. and Leach P.G.L.
Lie symmetries of generalized Burgers equations: application to boundary-value problems,
{\it J.~Engrg. Math.} {\bf 91} (2015), \href{https://doi.org/10.1007/s10440-008-9280-9}{165--176}; \href{https://arxiv.org/abs/1303.3548}{arXiv:1303.3548}.

\bibitem{Wang1988}
Wang X.Y.
Exact and explicit solitary wave solutions for the generalised Fisher equation,
\emph{Phys. Lett.~A} {\bf 131} (1988), \href{https://doi.org/10.1016/0375-9601(88)90027-8}{277--279}.

\bibitem{WhittakerWatson}
 Whittaker E.T. and Watson G.N. {\it A Course of Modern Analysis},
\href{https://doi.org/10.1017/CBO9780511608759}{Cambridge University Press}, Cambridge, 1996.

\bibitem{Winternitz1992}
Winternitz P. and Gazeau J.P. Allowed transformations and symmetry classes of variable coefficient Korteweg--de Vries equations, {\it Phys. Lett.~A} {\bf 167} (1992), \href{https://doi.org/10.1016/0375-9601(92)90199-V}{246--250}.

\bibitem{yehorchenko}
Yehorchenko I. Solutions of the system of d'Alembert and eikonal
equations, and classification of reductions of PDEs, {\it J.~Phys. Conf. Ser.} {\bf 621} (2015), \href{http://iopscience.iop.org/article/10.1088/1742-6596/621/1/012018}{012018, 6~pp.}

\bibitem{Zhdanov&Lahno1998}
Zhdanov R.Z. and Lahno V.I.
Conditional symmetry of a porous medium equation, {\it Phys.~D} {\bf 122} (1998), \href{https://doi.org/10.1016/S0167-2789(98)00191-2}{178--186}.

\bibitem{ZhdanovTsyfraPopovych1998}
 Zhdanov R.Z., Tsyfra I.M. and Popovych R.O. A precise definition of reduction of partial differential
equations, {\it J.~Math. Anal. Appl.} {\bf 238} (1999), \href{https://doi.org/10.1006/jmaa.1999.6511}{101--123}; \href{https://arxiv.org/abs/math-ph/0207023}{arXiv:math-ph/0207023}.

\end{thebibliography}
\end{document}